\documentclass[12pt,english]{article}
\pdfoutput=1
\usepackage{lmodern}

\usepackage[T1]{fontenc}
\usepackage[latin9]{inputenc}
\usepackage{geometry}
\usepackage{color}
\usepackage{babel}
\usepackage{cite}
\usepackage{mathtools}
\usepackage{amsmath}
\usepackage{amssymb}
\usepackage{graphicx}
\usepackage{esint}
\usepackage{bbm}
\usepackage[unicode=true,pdfusetitle,
 bookmarks=true,bookmarksnumbered=false,bookmarksopen=false,
 breaklinks=false,pdfborder={0 0 1},backref=false,colorlinks=true]
 {hyperref}
\hypersetup{pdftitle={Bulk Connectedness and Boundary Entanglement}, pdfauthor={Ning Bao, Grant N. Remmen},
 citecolor=black,linkcolor=black,urlcolor=black}
\makeatletter

 \@ifundefined{textcolor}{}
 {
   \definecolor{BLACK}{gray}{0}
   \definecolor{WHITE}{gray}{1}
   \definecolor{RED}{rgb}{1,0,0}
   \definecolor{GREEN}{rgb}{0,1,0}
   \definecolor{BLUE}{rgb}{0,0,1}
   \definecolor{CYAN}{cmyk}{1,0,0,0}
   \definecolor{MAGENTA}{cmyk}{0,1,0,0}
   \definecolor{YELLOW}{cmyk}{0,0,1,0}
 }
\usepackage{babel}
\makeatletter
\def\simgt{\mathrel{\lower2.5pt\vbox{\lineskip=0pt\baselineskip=0pt
           \hbox{$>$}\hbox{$\sim$}}}}
\def\simlt{\mathrel{\lower2.5pt\vbox{\lineskip=0pt\baselineskip=0pt
           \hbox{$<$}\hbox{$\sim$}}}}
\makeatother

\usepackage{braket}
\usepackage{amsthm}

\newtheorem{theorem}{Theorem}

\newtheorem{corollary}{Corollary}[theorem]

\newcommand{\be}{\begin{equation}}
\newcommand{\ee}{\end{equation}}
\newcommand{\bea}{\begin{eqnarray}}
\newcommand{\eea}{\end{eqnarray}}

\newcommand{\Ref}[1]{Ref.~\cite{#1}}
\newcommand{\Fig}[1]{Fig.~\ref{#1}}

\newcommand{\Sec}[1]{Sec.~\ref{#1}}

\usepackage{xcolor}

\begin{document}
\interfootnotelinepenalty=10000
\baselineskip=18pt

\hfill CALT-TH-2017-011
\hfill

\vspace{2cm}
\thispagestyle{empty}
\begin{center}
{\LARGE \bf
Bulk Connectedness and Boundary Entanglement
}\\
\bigskip\vspace{1cm}{
{\large Ning Bao$^{1,2,3}$ and Grant N. Remmen$^{1,3}$}
} \\[7mm]
{\it
$^1$Walter Burke Institute for Theoretical Physics \\[-1mm] California Institute of Technology, Pasadena, CA 91125, USA \\
$^2$Institute for Quantum Information and Matter,\\[-1mm] California Institute of Technology, Pasadena, CA 91125, USA \\
$^3$Center for Theoretical Physics and Department of Physics\\[-1mm]
    University of California, Berkeley, CA 94720, USA}
 \end{center}
\bigskip
\centerline{\large\bf Abstract}

\begin{quote} \small
We prove, for any state in a conformal field theory defined on a set of boundary manifolds with corresponding classical holographic bulk geometry, that for any bipartition of the boundary into two non-clopen sets, the density matrix cannot be a tensor product of the reduced density matrices on each region of the bipartition. 
In particular, there must be entanglement across the bipartition surface. We extend this no-go theorem to general, arbitrary partitions of the boundary manifolds into non-clopen parts, proving that the density matrix cannot be a tensor product. This result gives a necessary condition for states to potentially correspond to holographic duals. 
\end{quote}

\newpage
\tableofcontents

\newpage

\section{Introduction}
The question of which conformal field theory states can correspond to smooth, classical gravity duals must be answered in order to determine the limitations of AdS/CFT \cite{Maldacena:1997re, Witten:1998qj, Aharony:1999ti}. Approaches thus far have included constraints on the theories \cite{Heemskerk:2009pn} and on the entropies \cite{Hayden:2011ag, Bao:2015bfa}, though less work has been done to directly constrain the states themselves.

Folklore in the community has suggested that perhaps multipartite entanglement (entanglement that cannot be distilled into Bell pairs alone) may not be well suited for smooth gravity duals without some amount of bipartite entanglement, due to tension between classes of known multipartite entangled states and the Ryu-Takayanagi (RT) formula for entanglement entropy \cite{Ryu:2006bv, Hayden:2011ag, Bao:2015bfa}. Nevertheless, multiboundary wormhole geometries where the multipartite entanglement is provably existent are known, at least in three dimensions \cite{Balasubramanian:2014hda}. Again, however, less work has been done on constraining what classes of bipartite entangled states can have classical gravity duals.

In this work, we will focus on excluding certain patterns of entanglement from corresponding to smooth, classical gravity duals. All of our statements apply to the leading term in a $1/N$ expansion for a large-$N$ holographic CFT. In particular, we will demonstrate that, for conformal field theory (CFT) states defined on a set of boundaries, only those with nonzero entanglement entropy across any nontrivial subdivision are permitted. This result is consistent with the singularity in the energy-momentum tensor for factorizable states identified in \Ref{Czech:2012be}. In \Sec{sec:threeboundary}, we will gain intuition by considering a CFT defined on a set of three boundaries. We prove our main results in \Sec{sec:nogo}, first considering some requisite properties of bulk connectedness between sets of CFT boundaries and then interpreting and applying these results to study entanglement between those boundaries. In particular, we establish a set of necessary conditions on the boundary state for the existence of a holographic bulk geometry, generalized to any number of boundaries. We conclude and discuss implications of our result in \Sec{sec:discussion}.

\section{Three Boundary Case}\label{sec:threeboundary}

Let us begin by describing a bulk geometry that could hypothetically exist but that may be problematic from the perspective of information theory and holography. Consider three manifolds $A_1$, $A_2$, and $A_3$ on which we define a CFT and some pure state $|\psi\rangle$. Permit $A_1$ and $A_2$ to be connected through the bulk and similarly for $A_1$ and $A_3$, but keep $A_2$ and $A_3$ disconnected in the bulk; see Figs.~\ref{fig:unsplit} and \ref{fig:split}. Now partition $A_1$ into regions $A_1^{(2)}$ and $A_1^{(3)}$ (i.e., $A_1^{(2)} \cap A_1^{(3)}$ is empty and $A_1^{(2)} \cup A_1^{(3)} = A_1$) such that some geodesics beginning on $A_1^{(2)}$ enter the spacetime region connected with $A_2$ and some geodesics beginning on $A_1^{(3)}$ enter the region of spacetime connected with $A_3$. Holographically, the leading-order term in the entanglement entropy of a region is the area of the RT surface that subtends that region. Now construct the RT surface subtending $A_1^{(2)} \cup A_2$. This RT surface should have zero area, giving vanishing entanglement entropy. This would therefore be a nontrivial partition of the full Hilbert space that would have one side unentangled with the other. Writing the full pure state density matrix on the three CFT boundaries as $\rho=|\psi\rangle\langle\psi|$, we are disallowing the situation in which $\rho = \rho_{B_1}\otimes \rho_{B_2}$, where $\rho_{B_1}$ and $\rho_{B_2}$ are the reduced density matrices on the bipartition of $A_1 \cup A_2 \cup A_3$ into $B_1$ and $B_2$, for $B_1 = A_1^{(2)} \cup A_2$ and $B_3 = A_1^{(3)} \cup A_3$, respectively. In other words, $\rho$ cannot correspond to a separable pure state.

\begin{figure}[!tbp]
  \centering
  \begin{minipage}[b]{0.43\textwidth}
    \includegraphics[width=\textwidth]{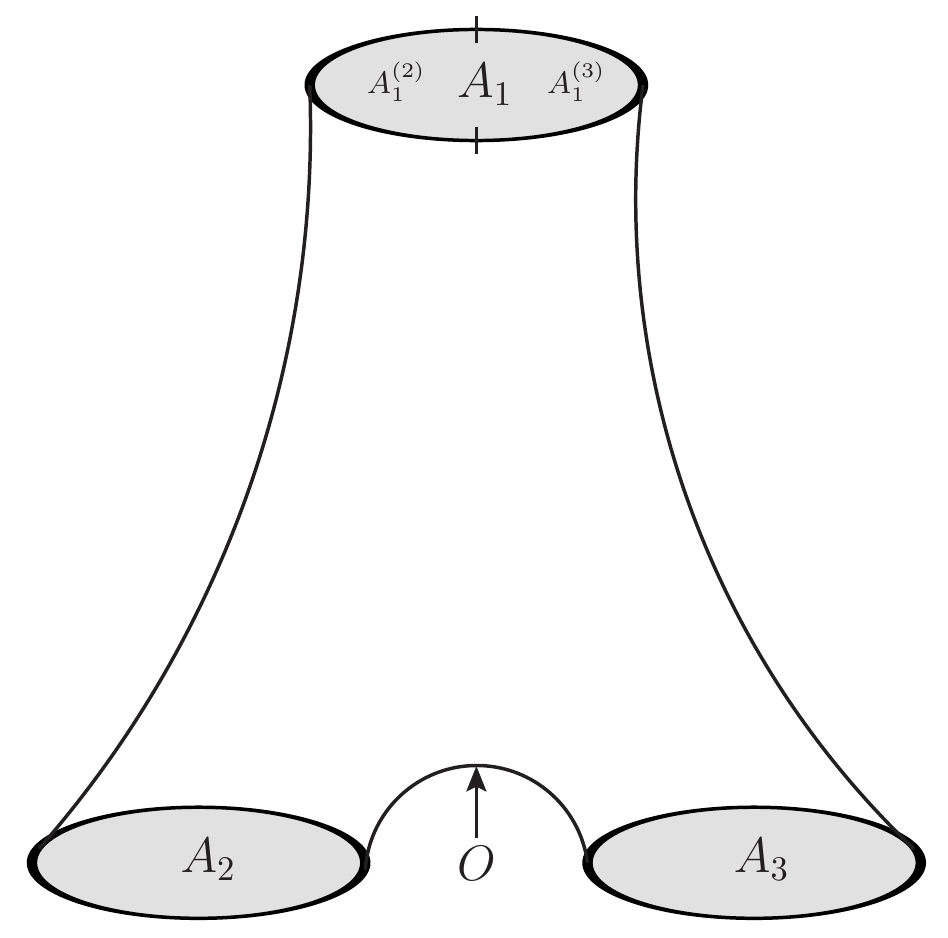}
    \caption{Situation in which $A_1$, $A_2$, and $A_3$ define the boundaries of a holographic geometry, in which there exist bulk paths connecting each pairwise combination of $A_1$, $A_2$, and $A_3$.\label{fig:unsplit}}
  \end{minipage}
  \hfill
  \begin{minipage}[b]{0.43\textwidth}
    \includegraphics[width=\textwidth]{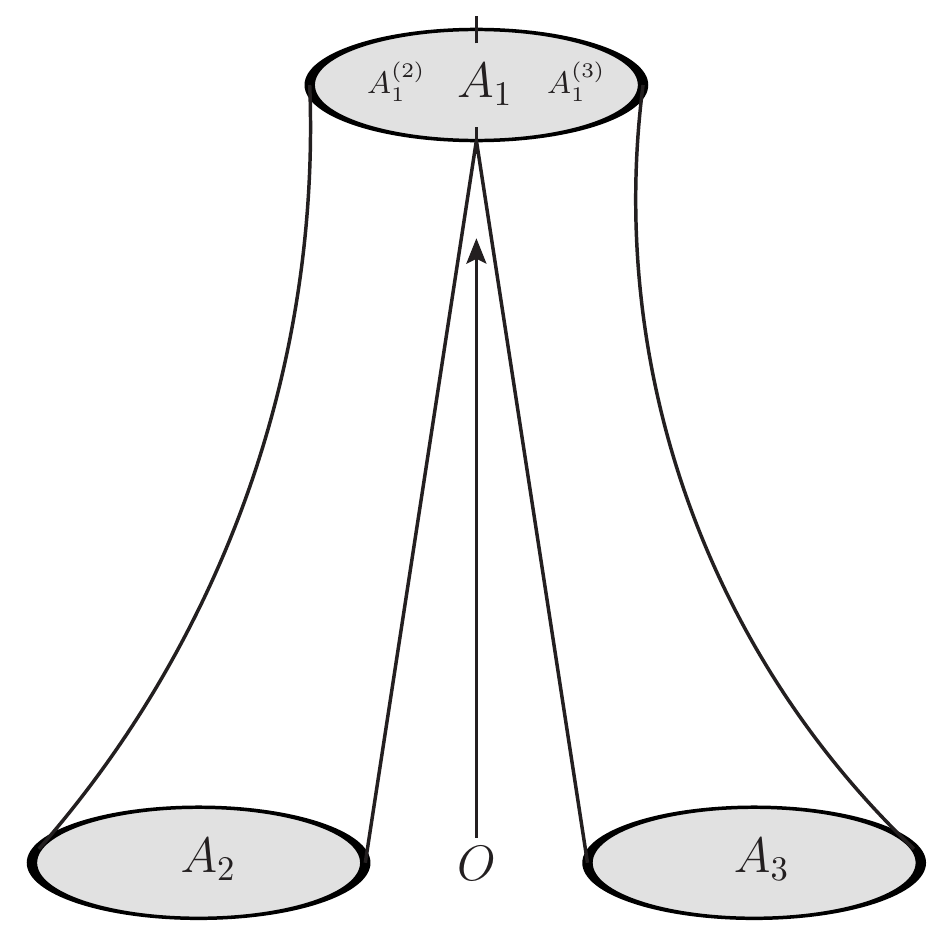}
    \caption{Situation in which $A_1$ is partitioned between $A_1^{(2)}$ and $A_1^{(3)}$, with $A_1^{(2)} \cup A_2 = B_1$ and $A_1^{(3)} \cup A_3 = B_2$ disconnected in the bulk.\newline \label{fig:split}}
  \end{minipage}
\end{figure}

The reason why a classical bulk spacetime corresponding to such a state would seem a priori surprising (and, as we will eventually see, problematic) from a holographic perspective is that there exist subregions $a_1$ and $a_2$ of $B_1 \cap A_1$ and $B_2 \cap A_1$, respectively, that can be chosen to be arbitrarily close to each other. If the entanglement entropy of $B_1$ is zero, however, this requires that there is no entanglement between $a_1$ and $a_2$, i.e., that the mutual information $I(a_1:a_2)=0$ for all $a_1$ and $a_2$. This follows from the fact that $B_1$ would be a pure state in this scenario and any reduced density matrix of a pure state has zero mutual information with any other density matrix that is not a different reduced density matrix of the original pure state. This situation would contradict the RT formula if $a_1$ and $a_2$ were sufficiently large and close together such that  the minimal surface subtending them both would not, in a conventional holographic geometry as in \Fig{fig:unsplit}, simply be the union of minimal surfaces subtending them separately. We now examine the geometry more closely, to gain intuition for what is causing this violation of expectations at the level of the RT formula. While in the remainder of this section we are simply developing intuition for the problem, we will demonstrate in \Sec{sec:nogo} that there is a genuine geometric pathology associated with CFT states of the form we wish to forbid.

Consider again the RT surface subtending $B_1$ in the geometry described above. This will correspond to the surface spanning the length of the spacetime region joining the two sides of the manifold depicted in \Fig{fig:unsplit}, touching $A_1$ at the intersection points of $B_1 \cap A_1$ and $B_2 \cap A_1$ and passing through $O$. This surface will in general have nonzero area in the bulk. Now consider the situation where $O$ is brought normally down to the CFT boundary $A_1$, as in \Fig{fig:split}. The RT surface as previously described still gives the entanglement entropy of $B_1$ via its bulk area. The bulk area of the RT surface, however, will now be zero, forcing there to be zero entanglement between $B_1$ and $B_2$. In other words, for the geometry depicted in \Fig{fig:split}, there is no surface in the bulk that is homologous to both $B_1$ and its complement $B_2$, so the RT surface cannot exist.

Motivated by this intuition, we conjecture that states of the form described, in which $B_1$ and $B_2$ are unentangled (i.e., in which the density matrix factorizes), for a CFT defined on the disjoint union of boundary manifolds $A_1$, $A_2$, and $A_3$, are forbidden from having a well-behaved holographic bulk geometry. In the next section, we will generalize this intuition into a precise statement about CFTs defined on an arbitrary number of boundaries and then prove the theorem.

\section{A No-Go Theorem for a Classical Holographic Bulk}\label{sec:nogo}

Let us now make some definitions. For a given time slice of a holographic CFT, we will denote as {\it the boundary} the collection of spacelike boundary manifolds on which the CFT is defined. We will call a surface $A$ a {\it complete boundary} if it is a connected component of the boundary and if $\partial A = 0$. Then the boundary $B$ is the disjoint union of complete boundary manifolds $A_1,...,A_n$. (We assume throughout that we work in a conformal frame such that $\partial A_i = 0$ on the slice of the spacetime we consider.) We will define a {\it bipartition} of $B$ as some identification of two disjoint subsets $B_1$ and $B_2$ of $B$ such that $B_1 \cup B_2 = B$. We will call such a bipartition {\it trivial} if $B_1$ and $B_2$ are each equal to the disjoint union of some subsets of the $A_i$, i.e., a trivial bipartition does not split any of the $A_i$ in two. Hence, a nontrivial bipartition has the property that neither $B_1$ nor $B_2$ is clopen, i.e., $\partial B_1$ and $\partial B_2$ are both nonempty.
By a {\it classical holographic bulk}, we will mean a continuous, classical spacetime manifold $\mathcal{M}$ corresponding to a CFT state, where for each spacelike slice $M$ of $\mathcal{M}$, $\partial M$ corresponds to the boundary on which the CFT is defined and, crucially, $\mathcal{M}$ is asymptotically anti-de Sitter (AdS), i.e., the metric on $M$ near each complete boundary approaches a spacelike slice of the AdS metric and, for any point in $M$, there exists a path within that slice connecting the point with one of the complete boundaries. Note that $\partial M$ is a spacelike collection of manifolds, while the boundary of the full spacetime, $\partial{\cal M}$, has Lorentzian signature and is topologically $\partial M \otimes {\mathbb R}$. Of particular importance for our arguments will be the requirement that $M$ satisfy the Hausdorff condition. While the Hausdorff condition is topological, it has a well known geometric consequence that we will use: geodesics cannot bifurcate \cite{Hajicek}. Hereafter, we will always restrict to some arbitrary spacelike slice $M$ of $\mathcal{M}$ and refer to $M$ as the classical holographic bulk. For simplicity, we will assume that the spacetime is approximately static, so that the RT formula accurately characterizes the entanglement entropy; it would be interesting to extend our results to more general spacetimes, which would require the Hubeny-Rangamani-Takayangi \cite{Hubeny:2007xt} formalism for the extremal surfaces. Given a CFT state with a classical holographic bulk, we will call boundary regions $A_1$ and $A_2$ {\it bulk (dis)connected} if there (does not) exist a path, contained entirely within the bulk, i.e., within the interior of $M$, connecting a point in $A_1$ with a point in $A_2$.

The discussion in \Sec{sec:threeboundary} then leads us to the proposition that, given a holographic CFT defined on three complete boundaries $A$, $B$, and $C$, there does not exist any classical holographic bulk in which $A$ is bulk connected with $B$ and $C$ but $B$ and $C$ are not bulk connected. Geometries violating this proposition will be pathological in a sense we will make precise. Accepting this proposition, we have a restriction on the CFT states themselves, in that for any state defined on $ABC$, any nontrivial bipartition must be entangled, i.e., the regions of any nontrivial bipartition must have nonvanishing mutual information. In the case of a pure state $\rho=\ket{\psi}\bra{\psi}$, this precludes $\rho$ from being separable; in other words, for any nontrivial bipartition of $ABC$ into regions $1$ and $2$, states in which $\rho = \rho_{1} \otimes \rho_{2}$ are forbidden from having a classical holographic bulk. 

We will now generalize these propositions to CFT boundaries with an arbitrary number of connected components, proving the results in two parts.

\begin{theorem}\label{thm:geometry} Given a CFT on a boundary consisting of an arbitrary number of complete boundaries $A_1,\ldots,A_n$ with a classical holographic bulk, for any two distinct points $x_1$ and $x_2$ in the same $A_i$, $x_1$ and $x_2$ must be bulk connected.
\end{theorem}

\begin{proof}
Suppose that the theorem is false, i.e., that there exists a classical holographic bulk $M$ in which, for some $x_1, x_2\in A_i$, we have that $x_1$ and $x_2$ are not bulk connected. Let us write as $M_1$ and $M_2$ the regions that are bulk connected to $x_1$ and $x_2$, respectively. That is, $M_1$ consists of all $x\in M$ such that $x$ is bulk connected to $x_1$; a similar definition holds for $M_2$. We write as $M_1^\mathrm{c}$ the complement of $M_1$ in the bulk, so by construction $M_1$ and $M_1^\mathrm{c}$ are bulk disconnected. Note that $M_1$ and $M_2$ must be disjoint in the bulk, since if there were a point $x\in M_1\cap M_2$, one could find a bulk path from $x_1$ to $x_2$ as the union of path $p_1$ from $x_1$ to $x$ and $p_2$ from $x_2$ to $x$, which is forbidden by the hypothesis that $x_1$ and $x_2$ are bulk disconnected. Thus, $M_2 \subset M_1^\mathrm{c}$, so since $M_2$ by hypothesis exists, $M_1^\mathrm{c}$ is nonempty.

Now, given a point $x_0\in A_i$, we can consider a (spacelike) geodesic $p(\lambda)\subset M$ extending into the bulk, where $x_0=p(0)$. We can enforce that the geodesic extends into the bulk by requiring that $|g_{\mu\nu}(\mathrm{d}p)/\mathrm{d}\lambda)^\mu (\mathrm{d}p/\mathrm{d}\lambda)^\nu|_{\lambda = 0}^{1/2} > |\gamma_{ab}(\mathrm{d}p)/\mathrm{d}\lambda)^a (\mathrm{d}p/\mathrm{d}\lambda)^b|_{\lambda = 0}^{1/2}$, where $g_{\mu\nu}$ is the bulk metric and $\gamma_{ab}$ is the induced metric on the boundary. Then it becomes well defined to say that $p(\lambda)$ enters $M_1$ or $M_1^\mathrm{c}$, but not both, depending on whether $p(\lambda)\in M_1$ or $M_1^\mathrm{c}$ for infinitesimal $\lambda > 0$.

Let us now denote as $U$ the space of initial data for spacelike geodesics in $M$ starting on $A_i$ and extending into the bulk. Importantly, since $M$ is a classical holographic bulk, it is continuous and asymptotically AdS. Geodesics in $M$ therefore asymptotically approach the geodesics of a spacelike slice of the AdS metric near each complete boundary. AdS geodesics do not bifurcate, since AdS satisfies the Hausdorff condition. By continuity of $M$, $U$ is connected. (In particular, the topology of $U$ is just $A_i \otimes \mathbb{R}^{D-2}$ for a bulk of spacetime dimension $D$.) For a given point $y\in U$, we can uniquely specify whether the geodesic $p_y$ to which it corresponds enters $M_1$ or $M_1^\mathrm{c}$; let $U_1 = \{y\in U | p_y \text{ enters } M_1\}$ and $U_2 = \{y\in U | p_y \text{ enters } M_1^\mathrm{c}\}$. By construction, $U=U_1 \cup U_2$ and, since a geodesic can only enter $M_1$ or $M_1^\mathrm{c}$, $U_1$ and $U_2$ are disjoint. Note that $U_1$ are $U_2$ are both nonempty, since there exist geodesics entering $M_1$ from $x_1\in A_i$ and entering $M_2\subset M_1^\mathrm{c}$ from $x_2 \in A_i$. Define $V$ as the boundary of $U_1$, which is also the boundary of $U_2$; we have $V=\bar{U}_1 \cap \bar{U}_2$. Since $U$ is connected and $U_1$ and $U_2$ are proper subsets, neither $U_1$ nor $U_2$ are clopen and hence $V$ is nonempty. By continuity of $M$, geodesics defined by initial data in $V$ must enter both $M_1$ and $M_1^\mathrm{c}$. Hence, for any $y\in V$, the geodesic $p_y$ bifurcates, in contradiction with the requirement that $M$ be Hausdorff. This contradiction completes the proof.
\end{proof}

The subtlety in proving Theorem~\ref{thm:geometry} is that, in general, bulk connectedness does not imply boundary connectedness. Given two points $x_1$ and $x_2$ in the same connected component of the boundary, one cannot in general take a path through the boundary from $x_1$ to $x_2$ and deform it into a bulk path connecting the two points. In essense, Theorem~\ref{thm:geometry} formally proves that whenever this deformation procedure fails, it necessarily implies the existence of bifurcate geodesics and hence a violation of the Hausdorff condition.

We note an immediate consequence of this result for the connectivity structure of the classical holographic bulk.

\begin{corollary}\label{cor:connected} Given a CFT on a boundary consisting of an arbitrary number of complete boundaries $A_1,\ldots,A_n$, there does not exist a classical holographic bulk in which $A_i$ is bulk connected with $A_j$ and $A_k$, but $A_j$ and $A_k$ are bulk disconnected.
\end{corollary}

\begin{proof}
We again proceed by contradiction, assuming that there exists a classical holographic bulk $M$ with boundary consisting of some complete boundary manifolds $A_1,\ldots,A_n$ such that, for some $A_{i,j,k}$, $A_i$ is bulk connected to $A_j$ and $A_k$, but $A_j$ and $A_k$ are bulk disconnected. 

We choose points $x_1$ and $x_2$ in $A_i$ such that $x_1$ is bulk connected to $A_j$ and $x_2$ is bulk connected to $A_k$. If $x_1 = x_2$ (hereafter, $x$), then we have chosen a point that is bulk connected with both $A_j$ and $A_k$. The space $U_x$ of initial data for (spacelike) geodesics originating on $x$ and extending into the bulk (within $M$) is just a subset of $U$ defined in Theorem~\ref{thm:geometry}, for which $U_1 \cap U_x$ and $U_2 \cap U_x$ are both nonempty. Since the topology of $U_x$ in an asymptotically AdS spacetime is $\mathbb{R}^{D-2}$ and hence $U_x$ is connected, it follows that neither $U_1 \cap U_x$ nor $U_2 \cap U_x$ are clopen. Thus, just as in Theorem~\ref{thm:geometry}, the set $V_x = V \cap U_x$ is nonempty and, by continuity of the spacetime, geodesics corresponding to initial data in $V_x$ must originate on $x$ and bifurcate, entering both $A_j$ and $A_k$. This situation violates the Hausdorff condition, in contradiction with the requirement that $M$ be Hausdorff. 

Similarly, if $x_1 \neq x_2$, then we have distinct points in the same boundary manifold that are not bulk connected, in which case Theorem~\ref{thm:geometry} immediately applies, forbidding this setup as a classical holographic bulk.
\end{proof}

Again, Corollary~\ref{cor:connected} formalizes the intuition that if one can draw a path from $A_j$ to $A_k$, possibly containing a segment passing through some boundary surface $A_i$, then one can deform this path into the bulk. While this is not possible for general non-Hausdorff manifolds, the deformation can only fail in spacetimes containing bifurcating geodesics and hence that violate the Hausdorff condition.

Let us now interpret Theorem~\ref{thm:geometry} in terms of the entanglement structure of the CFT state.

\begin{theorem}\label{thm:entanglement} Given a pure CFT state defined on a boundary that possesses a classical holographic bulk, the boundary regions $B_1$ and $B_2$ defined by any nontrivial bipartition must be entangled with each other. That is, if the state of the full boundary is pure, with density matrix $\rho$, we have $\rho \neq \rho_{B_1}\otimes\rho_{B_2}$.\end{theorem}

\begin{proof}

Suppose the theorem is false, so that for some CFT state with a classical holographic bulk there exists a nontrivial bipartition of the boundary into $B_1$ and $B_2$ such that $B_1$ and $B_2$ are unentangled. Since the bipartition is nontrivial, there must exist some complete boundary $A$ such that $A\cap B_1$ and $A \cap B_2$ are both nonempty.

In general, the RT surface associated with the entanglement between $B_1$ and $B_2$ is given by a bulk surface that is codimension-one within $M$ (codimension-two within $\mathcal{M}$) and is a disjoint union of complete codimension-one surfaces within each connected component of the bulk. By {\it complete} codimension-one surfaces, we mean surfaces $\Sigma$ such that, within a given connected component of $M$, any bulk path from $B_1$ to $B_2$ must pass through $\Sigma$. Among all complete codimension-one surfaces separating $B_1$ and $B_2$, the RT surface is the one with minimal area in $M$. The fact that $B_1$ and $B_2$ are unentangled implies that the RT surface associated with $B_1$ and $B_2$ has zero area and hence the surface does not exist in the bulk. Therefore, there must exist no bulk path from any point in $B_1$ to any point in $B_2$, since such a path would be required to pass through the nonexistent RT surface.

We therefore choose some $x_1 \in B_1 \cap A$ and $x_2 \in B_2 \cap A$. Since $B_1$ and $B_2$ are disjoint, $x_1$ and $x_2$ are distinct and, by the argument above, must be bulk disconnected. Theorem~\ref{thm:geometry} then applies, which implies that the geometry is not a classical holographic bulk. This contradiction completes the proof.
\end{proof}

We can actually generalize this result to partitions of the boundary into multiple parts. Let us define a {\it nontrivial partition} of the boundary $B$ as a partition of $B$ into $B_1,\ldots,B_n$, $B_i \cap B_j=\emptyset$, $\cup_i B_i=B$, such that $\partial B_i \neq 0,$ i.e., $B_i$ is not clopen for all $i$. We then have the following result.

\begin{corollary}
For any nontrivial partition of $B$ into $B_1,\ldots,B_n$, define $\rho_i$ to be the reduced density matrix on $B_i$. Then for a pure CFT state that corresponds to a classical holographic bulk, $\rho \neq \rho_i \otimes \cdots \otimes \rho_n$.
\end{corollary}

\begin{proof}
This result follows from iteratively applying Theorem~\ref{thm:entanglement}, in each successive step taking a nontrivial bipartition of a part of the previous bipartition.
\end{proof}

\section{Discussion}\label{sec:discussion}
We have proven that that there must be entanglement across any nontrivial bipartition of boundary CFT regions in any holographic theory via proof by contradiction with the Hausdorff condition in the dual spacetime. More specifically, our result holds for a state of any holographic theory that corresponds to a classical holographic bulk spacetime. In a related work, \Ref{Miyaji:2014mca} demonstrated that a certain class of conformally invariant boundary states have essentially zero real-space entanglement for {\it any} nonzero bipartition, furthermore suggesting that these states are dual to trivial bulks of zero spacetime volume. This is consistent in the context of the present work, since such states are forbidden by our no-go theorem and such spacetimes would have nonexistent bulk (and hence trivially would not qualify as classical holographic bulks per our definitions in \Sec{sec:nogo}). 

While we proved our result using the tools of bulk geometry, it is worthwhile understanding our conclusions from the perspective of the CFT alone. In particular, the Hadamard condition \cite{Kay:1988mu} on the CFT state $|\psi\rangle$ means that its ultraviolet properties are well approximated by the ground state; in the dual geometric description, the usual holographic intuition about bulk depth corresponding to renormalization group flow means that this requirement on $|\psi\rangle$ corresponds to the condition that the bulk be asymptotically AdS. By the Reeh-Schlieder theorem (see \Ref{Witten:2018zxz} and refs. therein), the vacuum state of a Poincar\'e-invariant quantum field theory is cyclic and separating for spacelike-separated open sets. As one can verify by acting on the vacuum with the Hamiltonian, the cyclicity property implies that the pure vacuum state cannot have a product state structure. In Theorem~\ref{thm:entanglement}, we reached this same conclusion geometrically, for a general excited state that is dual to an asymptotically AdS spacetime.

We will now recast our result in graph-theoretic language that can more succinctly encode this requisite entanglement structure, in particular showing that the only entanglement structures consistent with holography can all be represented as disjoint unions of complete graphs, as we will elaborate below.

Corollary~\ref{cor:connected} states that any pair of CFT boundaries must be connected to each other purely within the bulk if there exists a ``chain'' of bulk connections from one boundary to the next by which they can be linked. In other words, for a set of CFT boundaries $A_1,\ldots,A_n$ such that there is a bulk connection from $A_i$ to $A_{i+1}$ for all $i<n-1$, $A_1$ and $A_n$ must be bulk connected, i.e., there must exist a path from $A_1$ to $A_n$ that does not pass through any boundary manifold; it then follows that every pair of $A_i$ and $A_j$ must similarly be bulk connected. Let us represent a classical holographic bulk as a graph, with a vertex representing each boundary manifold and an edge between any two vertices that are bulk connected with each other. Consider a connected subgraph, containing vertices corresponding to bulk connected boundaries $A_1$ and $A_2$. If $A_2$ is bulk connected with another boundary $A_3 \neq A_1$, then Corollary~\ref{cor:connected} forces $A_1$ to be bulk connected with $A_3$ as well. In the graph piccture, this forces the existence of an edge between $A_1$ and $A_3$; continuing this reasoning enforces that the connected subgraph be a complete graph. As there can be CFT boundaries that are simply not associated by a chain of bulk connections, the bulk connectivity graph of the set of all CFT boundaries, for any state corresponding to a classical holographic bulk, must be a disconnected union of complete subgraphs.

Note that our bound does not imply that every subregion of the boundary is entangled with every other subregion. For example, consider a large bulk region (with boundary $A$) that contains two small black holes, each connected via a wormhole with a partner black hole in another asymptotically AdS region of spacetime (with boundaries $B$ and $C$, respectively). Then our results permit a scenario in which some small subsets $A_1$ and $A_2$ of the degrees of freedom on $A$ are in approximately thermofield double states with the degrees of freedom on $B$ and $C$, respectively, with $B$ and $C$ being unentangled. This setup is not forbidden by our result, since $A_1 \cup B$ and $A_2 \cup C$ do not make up the entire boundary and hence do not constitute a bipartition. In order to rule out scenarios in which $A_1 \cup B$ and $A_2 \cup C$ are unentangled, we would be required to include all of $A$ within either $A_1$ or $A_2$.

An intuitive way to state our result from the CFT perspective is that, for the forbidden states, the spacetime would necessarily need access to scales at which classical general relativity does not apply, for the geometry to be precisely defined all the way to the UV. There is, then, no consistent way to impose a cutoff in this theory. From the AdS perspective, states of the forbidden form necessarily have gravity-side pathologies that extend all the way to the boundary and thus prevent the geometry from being asymptotically AdS.

One can also extend this result to geometries with black hole singularities in the bulk, by considering the purification of the bulk geometry with singularities into a multiboundary pure state wormhole geometry. Then, one can take the nontrivial bipartition to be between some subset $A_1$ of the pre-purification boundary and its complement $A_2$ plus the post-purification boundaries $B$. As the pure state by our theorem cannot be of the form $\rho_{A_1}\otimes\rho_{A_2B}$, if at this point we perform a partial trace over $B$, the resulting mixed density matrix of the initial unpurified mixed state cannot be of the form $\rho_{A_1}\otimes\rho_{A_2}$. Thus, similar results would constrain the forms of mixed sates corresponding to one-sided black hole geometries. We note also that while, for mixed states, this is insufficient to prove nonseparability, it may be a step in that direction.

Thus, we see that the entanglement entropy is strongly constrained to be nonzero for any nontrivial bipartition of the CFT state on the boundary for which a classical holograhic bulk exists and, moreover, that the existence of a chain of bulk connections between a set of CFT boundaries requires the existence of a complete graph of bulk connections. This result complements the consequences for the known constraints on holographic entanglement entropy~\cite{Hayden:2011ag, Bao:2015nqa, Bao:2015bfa}, in that the set of states that it allows or excludes is independent of and of a different character than those allowed or excluded by the previously known constraints. Furthermore, our result gives geometric and topological justification to the suggestion of \Ref{Czech:2012be} that factorizability of the density matrix leads to a singularity in the energy-momentum tensor. This pathology suggests the possibility that, under certain requirements, our results forbidding density matrix factorization can hold in contexts outside of holography, though this lies outside the scope of the present work. In the future, it would be interesting to study what fraction of quantum states under some appropriate measure satisfies our criterion, to assess just how constraining a requirement this is on the space of holographic theories. Moreover, it would be interesting to consider the extension of this statement to an ER=EPR perspective \cite{Maldacena:2013xja, Remmen:2016wax}.

\begin{center} 
 {\bf Acknowledgments}
 \end{center}
 \noindent 

We thank the referees for their helpful suggestions. We also thank Raphael Bousso, Sean Carroll, Netta Engelhardt, and Mukund Rangamani for useful discussions and comments. This research was supported in part by DOE grant DE-SC0011632 and by the Gordon and Betty Moore Foundation through Grant 776 to the Caltech Moore Center for Theoretical Cosmology and Physics. N.B. was supported at Caltech by the DuBridge postdoctoral fellowship at the Walter Burke Institute for Theoretical Physics and is currently supported at the University of California, Berkeley by the National Science Foundation under grant number 82248-13067-44-PHPXH. G.N.R. was supported at Caltech by a Hertz Graduate Fellowship and a NSF Graduate Research Fellowship under Grant No.~DGE-1144469 and is currently supported at University of California, Berkeley by the Miller Institute for Basic Research in Science.

\bibliography{multiboundary}
\bibliographystyle{utphys}

\end{document}